\documentclass[twoside,11pt]{article}

\usepackage{jmlr2e}
\usepackage{amsmath}

\usepackage{xcolor}

\usepackage[english,ruled,vlined]{algorithm2e}

\usepackage[T1]{fontenc}
\usepackage{hyperref}

\usepackage{textcomp}

\newcommand{\modif}[1]{{#1}}

\ShortHeadings{Bouveyron, Latouche, and Mattei}{Exact Dimensionality Selection for Bayesian PCA}

\begin{document}

\title{\textsc{Exact Dimensionality Selection for Bayesian PCA}}

 \author{\name Charles Bouveyron \email \emph{charles.bouveyron@math.cnrs.fr} \\  
     \addr 
      Laboratoire J. A. Dieudonné, UMR CNRS 7135 \\ INRIA Epione team, Sophia
Antipolis \\ Universit\'e Côte d'Azur, Nice, France \\
          \name Pierre Latouche \email \emph{pierre.latouche@math.cnrs.fr} \\
    \addr Laboratoire MAP5, UMR CNRS 8145\\
     Universit\'e Paris Descartes -- Sorbonne Paris Cit\'e, France\\
   \name Pierre-Alexandre Mattei \email \emph{pima@itu.dk} \\
           \addr 
  Department of Computer Science\\
  IT University of Copenhagen, Denmark}

\maketitle

       \medskip

\begin{abstract}
We present a Bayesian model selection approach to estimate the intrinsic dimensionality of a high-dimensional dataset. To this end, we introduce a novel formulation of the probabilisitic principal component analysis model based on a normal-gamma prior distribution. In this context, we exhibit a closed-form expression of the marginal likelihood which allows to infer an optimal number of components. We also propose a heuristic based on the expected shape of the marginal likelihood curve in order to choose the hyperparameters. In non-asymptotic frameworks, we show on simulated data that this exact dimensionality selection approach is competitive with both Bayesian and frequentist state-of-the-art methods.
\end{abstract}


\section{Introduction}

The computer age is characterized by a surge of multivariate data, which are often difficult to explore or describe. A natural way to deal with such datasets is to reduce their dimensionality in a interpretable way, trying not to loose too much information. Accordingly, a wide range of dimension-reduction techniques have been developed over the years. Principal component analysis (PCA), perhaps the earliest of these techniques, remains today one of the most widely used \citep{jolliffe2016principal}.
Introduced by \cite{pearson1901} and rediscovered by \cite{hotelling1933} in the beginning of the twentieth century, PCA has had indeed an ubiquitous role in statistical analysis since the introduction of electronic computation in the 1950s. Recent exemples include climate research \citep{hannachi2006search}, genome-wide expression studies \citep{ringner2008}, massive text mining \citep{zhang2011}, and deep learning \citep{pcanet}. For a more exhaustive overview of past applications of PCA, we defer the reader to the monograph of \cite{jolliffe2002principal} or the recent review paper of \cite{jolliffe2016principal}.

Specifically, PCA consists in a simple procedure: the practitioner orthogonally projects his multivariate data on a space spanned by the  eigenvectors associated with the largest eigenvalues of the empirical covariance matrix. The dimension of the representation learnt in this way is simply the number of eigenvectors -- called principal components (PCs) -- kept for the projection. However, it may come as a surprise that in spite of the popularity of this method, no authoritative solution has been widely accepted for choosing how many PCs should be computed. Common practice is to choose this dimension by considering the eigenvalues scree of the sample covariance matrix. This ad-hoc technique, popularized by \cite{cattell1966scree}, has been largely modified and perfected over the last fifty years \citep{jackson1993stopping,zhu2006automatic}, and is often chosen when PCA is used as a building block within a larger algorithmic framework -- see e.g. \cite{bouveyron2007high} for an example in cluster analysis or \cite{evangelopoulos2012latent} in latent semantic analysis. However, more refined approaches have also been developed. Earlier works were based on hypothesis testing \cite[Section 6.1.4]{jolliffe2002principal}. Cross-validation, suggested by \cite{wold1978cross} and developed over the years \citep{bro2008cross}, is known to be effective in a wide variety of settings \citep{josse2012selecting}. Another fruitful line of work follows the seminal article of \cite{tipping1999}, who recast PCA as a simple inferential problem. Their model, called probabilistic PCA (PPCA), led to several model-based methods for dimensionality selection, both from frequentist \citep{ulfarsson2008dimension,bouveyron2011intrinsic,passemier} and Bayesian \citep{bishop1999bayesian,minka2000automatic,hoyle2008automatic,sobczyk2018} perspectives.

Most of the aforementioned methods are based on asymptotic considerations. However, it was recently proven that, in an asymptotic framework, hard thresholding the eigenvalues surprisingly suffices to provide an optimal dimensionality \citep{gavish2014optimal}. Thus, the path to more efficient schemes for finding the number of PCs goes through the study of non-asymptotic criteria, which have been overlooked in the past. A natural non-asymptotic answer is provided by exact Bayesian model selection, which was previously used at the price of computationally expensive Markov chain Monte Carlo (MCMC) sampling \citep{hoff2007model}. We present here a prior structure based on the PPCA model that allows us to exhibit a closed-form expression of the marginal likelihood, leading to an efficient algorithm that selects the number of PCs without any asymptotic assumption. Specifically, we rely on a normal prior distribution over the loading matrix and a gamma prior distribution over the noise variance. Imposing a simple constraint on the hyperparameters of the respective distributions, we show that this allows the data to marginally follow a generalized Laplace distribution, leading to an efficient closed-form computation of the marginal likelihood. We also propose a heuristic based on the expected shape of the marginal likelihood curve in order to choose hyperparameters. With simulated data, we demonstrate that our approach is competitive compared to state-of-the-art methods, especially in non asymptotic settings and with less observations than variables. This setting is at the core of many practical problems, such as genomics and chemometrics.

In Section 2, we briefly review PPCA and present several dimensionality selection techniques based on this model. The new normal-gamma prior is presented in Section 3 together with a derivation of the closed-form expression of the marginal likelihood. A heuristic to choose hyperparameters is also presented. Numerical experiments are provided in Section 4.

\section{Choosing the intrinsic dimension in probabilistic PCA}

Let us assume that a centered independent and identically distributed (i.i.d.) sample $\mathbf{x}_1,...,\mathbf{x}_n \in \mathbb{R}^p$ is observed that we aim at projecting onto a $d$-dimensional subspace while retaining as much variance as possible. All the observations are stored in the $n\times p$ matrix $\mathbf{X}=(\mathbf{x}_1,...,\mathbf{x}_n)^T$.

\subsection{Probabilistic PCA}

The PPCA model $\mathcal{M}_d$ assumes that, for all $i \in \{1,...,n\}$, each observation is driven by the following generative model
\begin{equation} \label{modelePPCA}
\mathbf{x}_i = \mathbf{W} \mathbf{y}_i + \boldsymbol{\varepsilon}_i,
\end{equation} where  $\mathbf{y}_i \sim \mathcal{N}(0,\mathbf{I}_d)$ is a low-dimensional Gaussian latent vector, $\mathbf{W}$ is a $p \times d$ parameter matrix called the \emph{loading matrix} and  $\boldsymbol{\varepsilon}_i \sim \mathcal{N}(0, \sigma^2 \mathbf{I}_p)$ is a Gaussian noise term.

This model is an instance of factor analysis and was first introduced by \cite{lawley1953}. \cite{tipping1999} then presented a thorough study of this model. In particular, expanding a result of \cite{theobald}, they proved that this generative model is indeed equivalent to PCA in the sense that the principal components of $\mathbf{X}$ can be retrieved using the maximum likelihood (ML) estimator $\mathbf{W}_{\textup{ML}}$ of $\mathbf{W}$. More specifically, if $\mathbf{A}$ is the $p \times d$ matrix of ordered principal eigenvectors of $\mathbf{X}^T\mathbf{X}$ and if $\boldsymbol{\Lambda}$ is the $d \times d$ diagonal matrix with corresponding eigenvalues, we have
\begin{equation} \label{ML}
\mathbf{W}_{\textup{ML}} = \mathbf{A}(\boldsymbol{\Lambda}-\sigma^2\mathbf{I}_d)^{1/2}\mathbf{R},
\end{equation}
where $\mathbf{R}$ is an arbitrary orthogonal matrix.

Under this sound probabilistic framework, dimension selection can be recast as a \emph{model selection problem}, for which standard techniques are available. We review a few important ones in the next subsection. 

\subsection{Model selection for PPCA}

The problem of finding an appropriate dimension can be seen as choosing a "best model" within a family of models $(\mathcal{M}_d)_{d \in \{1,...,p-1\}}$. A first popular approach would be to use likelihood penalization, leading to the choice
$$ d^* \in \textup{argmax}_{d \in \{1,...,p-1\}} \{ \log p(\mathbf{X}|\mathbf{W}_{\textup{ML}},\mathbf{\sigma}_{\textup{ML}},\mathcal{M}_d)-\textup{pen}(d) \},$$
where pen is a penalty which grows with $d$. These methods include the popular Akaike information criterion (AIC, \citealp{akaike1974new}), the Bayesian information criterion (BIC, \citealp{schwarz1978estimating}), or other refined approaches \citep{bai2002determining}. However, their merits are mainly asymptotic, and our main interest in this paper is to investigate non-asymptotic scenarios. While the penalty term is usually necessary to avoid selecting the largest model, under a constrained PPCA model, called isotropic PPCA, \cite{bouveyron2011intrinsic} proved that regular maximum likelihood was suprinsingly consistent. While the theoretical optimality of this method is also asymptotic, the fact that it directly maximizes a likelihood criterion which is not derived based on asymptotic considerations makes it of particular interest within the scope of this paper.

Another interesting set of techniques non-asymtotic in essence is Bayesian model selection \citep{kass1995}. \modif{Such approaches require the (approximate) computation of the marginal likelihood of Bayesian versions of the PPCA model. However, the usual Bayesian information criterion (BIC) approximation fails to approximate the marginal likelihood in the case of PPCA because of violated regularity conditions. To grasp the origin of these violations, consider the case where the true intrinsic dimensionality is one -- so that the data lives close to a  ``true line''. It is then possible to find a continuously infinite set of 2-dimensional planes that all contain this true line, leading to the non-invertibility of the Fisher information matrix of the PPCA model for $d=2$ in some parts of the parameter space, and to the failure of the Laplace approximation that underlies the BIC. More details on these problems for the very close factor analysis model can be found in \citet{drton2017}. Note that, even though the BIC provides a poor approximation of the marginal likelihood of a PPCA model, it can asymptotically lead to consistent model selection in a variety of settings \citep{bai2018}.
A more refined approach than the BIC was proposed by \cite{minka2000automatic} who derived a Laplace approximation of the marginal likelihood, that involves the eigenvalues of the covariance matrix. This technique, albeit asymptotic \modif{and subject to the same regularity violations}, has been proven empirically efficient in several small-sample scenarios.}

Another interesting framework considered in the literature is the case where both $n$ and $p$ grow to infinity. Several consistent estimators have been proposed, both from a penalization point of view \citep{bai2002determining,passemier,bai2018}, using Stein's unbiased risk estimator \citep{ulfarsson2008dimension} or in a Bayesian context \citep{hoyle2008automatic,sobczyk2018}. While these high-dimensional scenarios are of growing importance, they fall beyond the scope of this paper, which is focused on the non-asymptotic setting (with potentially fewer observations than variables), for which very few automatic dimension selection methods are available.

\section{Exact dimensionality selection for PPCA under a normal-gamma prior}

In this section, we present a prior structure that leads to a closed-form expression for the marginal likelihood of PPCA.

\subsection{The model}

We consider the regular PPCA model already defined in \eqref{modelePPCA},
\begin{equation*} 
\forall i \in \{1,...,n\}, \; \mathbf{x}_i = \mathbf{W} \mathbf{y}_i + \boldsymbol{\varepsilon}_i,
\end{equation*} 
where  $\mathbf{y}_i \sim \mathcal{N}(0,\mathbf{I}_d)$, $\mathbf{W}$ is a $p \times d$ parameter matrix, and $\boldsymbol{\varepsilon}_i \sim \mathcal{N}(0, \sigma^2 \mathbf{I}_p)$. We rely on a Gaussian prior distribution over the loading matrix $\mathbf{W}$ and a gamma prior distribution over the noise variance $\sigma^2$. Specifically, we use a gamma prior $\sigma^2 \sim \textup{Gamma}(a,b)$ with hyperparameters $a>0$ and $b>0$ together with i.i.d. Gaussian priors \modif{for the entries of the loading matrix} $w_{jk} \sim \mathcal{N}(0,\phi^{-1})$ for $j \in \{1,...,p\}$ and $k \in \{1,...,d\}$ with some precision hyperparameter $\phi >0$. 

Within the framework of Bayesian model uncertainty \citep{kass1995}, the posterior probabilities of models can be written as, for all $d \in \{1,...,p\}$,
\begin{equation} p(\mathcal{M}_d|\mathbf{X},a,b,\phi) \propto p(\mathbf{X}|a,b,\phi,\mathcal{M}_d) p(\mathcal{M}_d), \end{equation}
where $$p(\mathbf{X}|a,b,\phi,\mathcal{M}_d) =\prod_{i=1}^n \int_{\mathbb{R}^{d \times p} \times \mathbb{R}^+ } p(\mathbf{x}_i|\mathbf{W},\sigma,\mathcal{M}_d) p(\mathbf{W}|\phi) p (\sigma | a,b) d\mathbf{W}d\sigma,$$ is the \emph{marginal likelihood} of the data under conditional independence \citep{kass1989}. Note that this expression also involves model prior probabilities -- in this paper, we will simply consider a uniform prior $$\forall d \in \{1,...,p\},  \; p(\mathcal{M}_d) \propto 1.$$

Computing the high-dimensional integral of the marginal likelihood usually comes at the price of various approximations \citep{bishop1999bayesian,minka2000automatic,hoyle2008automatic} or expensive sampling \citep{hoff2007model}.
However, with our specific choice of priors, and imposing a constraint on their respective hyperparameters, we  obtain a closed-form expression for the marginal likelihood. 

\begin{theorem}
Let $d \in \{1,...,p\}$. Under the normal-gamma prior with $b=\phi/2$, the log-marginal likelihood of model $\mathcal{M}_d$ is given by 
\begin{equation} \label{maintheo}
\begin{aligned}
\log  p(\mathbf{X}|a,\phi,\mathcal{M}_d)&= \sum_{i=1}^n\log p(\mathbf{x}_i|a,\phi,\mathcal{M}_d) \\
&= - n \log 2 -\frac{np}{2}\log(2\pi)-\frac{np}{2}\log(2\phi^{-1})-n \log \Gamma(a+d/2) \\
&\quad + (a+\frac{d-p}{2}) \sum_{i=1}^n\log(\frac{\sqrt{\phi} ||\mathbf{x}_i||_2}{2} )
+ \sum_{i=1}^{n}\log K_{a+(d-p)/2}(\sqrt{\phi}||\mathbf{x}_i||_2),
\end{aligned}
\end{equation}
where $K_{\nu}$ is the modified Bessel function of the second kind of order $\nu\in \mathbb{R}$.
\end{theorem}
A detailed proof of this theorem is given in the next subsection. \modif{Note that the modified Bessel function $K_{\nu}$  can be evaluated using most statistical computing software. For instance, we used in our experiments the R package \texttt{Bessel} \citep{maechler}.}

To the best of our knowledge, this result is the first computation of the marginal likelihood of a PPCA model. It is worth mentioning that, in a slightly different context, \cite{ando2009} also derived the marginal likelihood of a factor analysis model, with Student factors. Similarly, \cite{bouveyron2016} derived the exact marginal likelihood of the noiseless PPCA model, in order to obtain sparse PCs. 

 It is worth noticing that the use of a gamma prior for a variance parameter is rather peculiar. Indeed, most Bayesian hierarchical models choose \emph{inverse-gamma} priors for variances. This choice is often motivated by its conjugacy properties (see e.g. \citealp{george1993}, for a linear regression example or \citealp{murphy}, in a wider setting). The derivation provided in the next subsection notably explains why this gamma prior over $\sigma^2$ actually arises naturally.

\modif{Regarding the loading matrix, Gaussian priors have been extensively used in the past \citep{bishop1999bayesian,archambeau2009,nakajima15,bouveyron2016}. They also correspond to the usual prior choice for probabilistic matrix factorization \citep{mnih2008}. Since the number of ``free parameters'' of the loading matrix is $dp-d(d-1)/2$ (see e.g. \citealp[Appendix B]{sobczyk2018}), such Gaussian priors that sample independently $dp$ values can be seen as too overparametrised. Consequently, several methods build priors on parameter spaces of smaller dimensions. Indeed, both \citet{minka2000automatic} and \cite{hoff2007model} place priors on the singular values decomposition of $\mathbf{W}$ rather than directly on $\mathbf{W}$. These approaches allow to maintain the identifiability of the PPCA model. On the other hand, overparametrised priors, like the one considered in this paper, lead to more complex and less interpretable posteriors, but have particularly interesting dimensionality selection properties, both in practice \citep{bishop1999bayesian,archambeau2009,knowles2011} and in theory \citep{nakajima15}. Note also that \citet{ando2009} as well as \citet{bouveyron2016}, who derived closed-forms of marginal likelihoods of related models, relied on such overparametrised priors.}

When it comes to hyperparameters, while \cite{hoff2007model} uses empirical Bayes heuristics, \citet{minka2000automatic} avoids hyperparameter specification by relying on the asymptotics of the Laplace approximation. Similarly, the BIC-based approach of \citet{sobczyk2018} is prior-independent. Being prior-dependent our approach would allow conversely to tune hyperparameters using prior knowledge (for example, prior knowledge on $\sigma$ may be known by assessing PCA reconstruction quality). We also propose an automatic empirical Bayes way of choosing these hyperparameters in Section \ref{ss:hyp}.

\subsection{Derivation of the marginal likelihood}

We begin by shortly reviewing the generalized Laplace distribution, which will prove to be key within the PPCA framework. This distribution was introduced by \citet[p. 257]{kotz2001}. For a more detailed overview, see \citet{kozubowski2013}.

\begin{definition} A random variable $\mathbf{z} \in \mathbb{R}^p$ is said to have a \textbf{multivariate generalized asymmetric Laplace distribution} with parameters $s>0, \boldsymbol{\mu} \in \mathbb{R}^p$ and $\mathbf{\Sigma} \in \mathcal{S}_p^+$ if its characteristic function is $$\forall \mathbf{u} \in \mathbb{R}^p, \; \phi_{\textup{GAL}_p(\mathbf{\Sigma}, \boldsymbol{\mu},s)}(\mathbf{u})=\left(\frac{1}{1+\frac{1}{2} \mathbf{u}^T\mathbf{\Sigma}\mathbf{u} - i \boldsymbol{\mu}^T\mathbf{u}}\right)^s.$$\end{definition}

When $\boldsymbol{\mu}=0$, the generalized Laplace distribution is elliptically contoured and is referred to as the \emph{symmetric} generalized Laplace distribution. The elementary properties of the generalized Laplace distribution are discussed by \citet{kozubowski2013}. We list the ones that we consider in the proof of Theorem 1.

\begin{proposition} If $\mathbf{z} \sim \textup{GAL}_p(\mathbf{\Sigma}, \boldsymbol{\mu},s)$, we have $\mathbb{E}(\mathbf{z})=s\boldsymbol{\mu}$ and $\textup{Cov}(\mathbf{z})=s(\mathbf{\Sigma} +\boldsymbol{\mu}\boldsymbol{\mu}^T )$. Moreover, if  $\mathbf{\Sigma}$ is positive definite, the density of $\mathbf{z}$ is given by \begin{equation}
\forall \mathbf{x} \in \mathbb{R}^p, \;  f_\mathbf{z}(\mathbf{x}) = \frac{2 e^{\boldsymbol{\mu}^T\mathbf{\Sigma}^{-1}\mathbf{x}}}{(2\pi)^{p/2}\Gamma(s)\sqrt{\det{\boldsymbol{\Sigma}}}} \left( \frac{Q_{\mathbf{\Sigma}}(\mathbf{x})}{C(\boldsymbol{\Sigma},\boldsymbol{\mu})}\right)^{s-p/2} K_{s-p/2} \left(Q_{\mathbf{\Sigma}}(\mathbf{x}) C(\boldsymbol{\Sigma},\boldsymbol{\mu})\right), \label{density} \end{equation} where $ Q_{\mathbf{\Sigma}}(\mathbf{x})= \sqrt{\mathbf{x}^T\mathbf{\Sigma}^{-1}\mathbf{x}}$ and $C(\boldsymbol{\Sigma},\boldsymbol{\mu})=\sqrt{2 + \boldsymbol{\mu}^T\mathbf{\Sigma}^{-1}\boldsymbol{\mu}}$.
\end{proposition}

\begin{proposition} \label{prop:summation}
Let $s_1,s_2>0, \boldsymbol{\mu} \in \mathbb{R}^p$ and $\mathbf{\Sigma} \in \mathcal{S}_p^+$. If $\mathbf{z}_1\sim \textup{GAL}_p(\mathbf{\Sigma}, \boldsymbol{\mu},s_1)$ and $\mathbf{z}_2\sim \textup{GAL}_p(\mathbf{\Sigma}, \boldsymbol{\mu},s_2)$ are independant random variables, then \begin{equation} \mathbf{z}_1+\mathbf{z}_2\sim \textup{GAL}_p(\mathbf{\Sigma}, \boldsymbol{\mu},s_1+s_2).
\end{equation}
\end{proposition}
This proposition is a direct consequence of the expression of the characteristic function of the generalized Laplace distribution.

 Another appealing property of  the multivariate generalized Laplace distribution is that it can be interpreted as an infinite scale mixture of Gaussians with gamma mixing distribution (a property called \emph{Gauss-Laplace representation} by \citealp{ding2017}).

\begin{proposition}[Generalized Gauss-Laplace representation] \label{prop:transmut}
Let $s>0$ and $\mathbf{\Sigma} \in \mathcal{S}_p^+$. If $u \sim \textup{Gamma}(s,1)$ and $\mathbf{x} \sim \mathcal{N}(0,\mathbf{\Sigma})$ is independent of $u$, we have\begin{equation}  \sqrt{u}\mathbf{x}\sim \textup{GAL}_p(\mathbf{\Sigma}, 0,s). \label{normalmean}
\end{equation}
\end{proposition}
For a proof of this result, see \citet[Chapter 6]{kotz2001}. 

To prove Theorem 1, we first study the marginal distribution of the signal term. Following \cite{mattei2017multiplying}, we can state the following lemma.

\begin{theorem} 
	Let $\mathbf{W}$ be a $p \times d$ random matrix with i.i.d. columns following a $\mathcal{N}(0,\phi^{-1}\mathbf{I}_p)$ distribution, $\mathbf{y}\sim \mathcal{N}(0,\mathbf{I}_d)$ be a Gaussian vector independent from $\mathbf{W}$. We obtain \begin{equation}
	\mathbf{Wy} \sim  \textup{GAL}_p(2 \phi^{-1}\mathbf{I}_p, 0,d/2).
	\end{equation}
	
\end{theorem}

\begin{lemma} 
Let $\mathbf{W}$ be a $p \times d$ random matrix with i.i.d. columns following a $\mathcal{N}(0,\phi^{-1}\mathbf{I}_p)$ distribution, $\mathbf{y}\sim \mathcal{N}(0,\mathbf{I}_d)$ be a Gaussian vector independent from $\mathbf{W}$. We obtain \begin{equation}
\mathbf{Wy} \sim  \textup{GAL}_p(2 \phi^{-1}\mathbf{I}_p, 0,d/2).
\end{equation}

\end{lemma}
\begin{proof}
For each $k \in \{1,...,d\}$ let $\mathbf{w}_k$ be the $k$-th column of $\mathbf{W}$, $u_k=y_k^2$ and $\boldsymbol{\xi}_k=y_k \mathbf{w}_k $. To prove the lemma, we demonstrate that $\boldsymbol{\xi}_1,...,\boldsymbol{\xi}_d$ follow a $ \textup{GAL}$ distribution and use the decomposition $$\mathbf{Wy} =\sum_{k=1}^d \boldsymbol{\xi}_k.$$

Let $k \in \{1,...,d\}$. Since $\mathbf{y}$ is standard Gaussian, $u_k=y_k^2$ follows a $\chi^2(1)$ distribution, or equivalently a $\textup{Gamma}(1/2,1/2)$ distribution. Therefore, $u_k/2 \sim \textup{Gamma}(1/2,1)$. Moreover, note that $ \sqrt{u_k} \mathbf{w}_k=|y_k|\mathbf{w}_k=y_k \textup{sign}(y_k)\mathbf{w}_k\buildrel d \over =y_k\mathbf{w}_k$ since $|y_k|$ and $\textup{sign}(y_k)$ are independent and $\textup{sign}(y_k)\mathbf{w}_k\buildrel d\over =\mathbf{w}_k$. Therefore, according to Proposition \ref{prop:transmut}, we have
$$ \boldsymbol{\xi}_k=\sqrt{\frac{u_k}{2}} \sqrt{2}\mathbf{w}_k \sim  \textup{GAL}_p(2\phi^{-1}\mathbf{I}_p,0,1/2).$$

Since $\boldsymbol{\xi}_1,...,\boldsymbol{\xi}_d$ are i.i.d. and following a $\textup{GAL}_p(2\phi^{-1}\mathbf{I}_p,0,1/2)$ distribution, we can use Proposition \ref{prop:summation} to conclude that $$\mathbf{Wy} =\sum_{k=1}^d \boldsymbol{\xi}_k\sim  \textup{GAL}_p(2 \phi^{-1}\mathbf{I}_p,0,d/2).$$ \end{proof}

We now focus on the second term of \eqref{modelePPCA} involving the noise vector.

\begin{lemma}
Let $\boldsymbol{\varepsilon}_i|\sigma^{2} \sim \mathcal{N}(0, \sigma^{2}\mathbf{I}_p)$ and $\sigma^{2} \sim \mathrm{Gamma}(a, b)$ then
$$\boldsymbol{\varepsilon}_i \sim \textup{GAL}_p\left( b^{-1}\textbf{I}_p,0,a \right).$$
\end{lemma}

\begin{proof}
  Again, the Gauss-Laplace representation is leveraged. Indeed, the noise can be written as $$\boldsymbol{\varepsilon}_i = \sqrt{b\sigma^2} \mathbf{e}_i,$$
where $\mathbf{e}_i \sim \mathcal{N}(0,b^{-1}\mathbf{I}_p)$. Therefore, the Gauss-Laplace representation allows to conclude.
\end{proof}

Now that we have proved that both the signal and the noise term follow marginally a generalized Laplace distribution, we use Proposition \ref{prop:summation} which ensures that, assuming $b=\phi/2$, the sum of the two generalized Laplace random vectors is a generalized Laplace random vector:
\begin{equation}
\mathbf{x}_i\sim \textup{GAL}_p(2\phi^{-1}\mathbf{I}_p, 0, a+ d/2).
\end{equation}

Using the expression of the density of the generalized Laplace distribution, we eventually end up with the closed-form expression of the marginal likelihood of Theorem 1.

\subsection{Choosing hyperparameters}
\label{ss:hyp}
To obtain a closed-form expression of the marginal likelihood, we have shown that it is sufficient to assume that $b=\phi/2$. \modif{This constraint, which appeared quite arbitrarily for the sake of mathematical convenience, has the benefit of not being too limiting. Indeed, since the other parameter of the gamma prior is unconstrained, the prior variances of the loading matrix and the noise variance remain untied by the constraint, and may be chosen independently.}
Two hyperparameters remain henceforth to be tuned: the shape parameter of the gamma prior $a$ and the precision hyperparameter $\phi$. We developed data-driven heuristics for this purpose.

A first observation is that, when $d$ grows, $\sigma$ is expected to decay because the signal part of the model can be more expressive. This prior information can be distilled into the model by roughly centering the gamma priors on estimates of $\hat{\sigma}$ {(note that this rationale is close to the one of \citealp{hoff2007model}).} More precisely, our heuristic is to choose $a$ such that $\mathbb{E}(\sigma) \propto \hat{\sigma}$ for each $d$. In order for $\phi$ to control the diffusiveness of both the loading matrix and the variance, we specifically made the choice $a=\hat{\sigma}^2/\phi$. In our experiments, we chose the ML estimator $\hat{\sigma}=\sigma_\textup{ML}$ (which is the mean of the $p-d$ smallest eigenvalues of the covariance matrix, see \citealp{tipping1999}) but more complex estimates may be considered \citep{passemier}.

Regarding the remaining parameter $\phi$, we propose a heuristic based on the following statements which can be made regarding the problem of  dimension selection:
\begin{itemize}
\item overestimation of $d$ should be preferred to underestimation since loosing some information is much more damageable than having a representation not parsimonious enough,
\item consequently, the marginal likelihood curve as a function of the dimension should have two distinct phases: a first one when "signal dimensions" are added (before the true value of $d$), and a second one, when "noise dimensions" are added.
\end{itemize}
Thus, we built a simple heuristic criterion to judge the relevance of a choice of $\phi$ by the shape of the marginal likelihood curve. First, if the slope of the first part of the curve (before the maximum) is lower than the slope of the second part, this means that this choice leads to underestimation and is therefore discarded. Second, the criterion is equal to the discrete second derivative of the marginal likelihood curve evaluated at the maximum, in order to select a hyperparameter leading to a strong distinction between the two phases. This criterion is eventually maximized over a grid of values of $\phi$. \modif{When there is no maximum, we set the heuristic criterion to $-\infty$: this is equivalent to putting zero prior mass on the two extreme models of the curve, which is consistent with the idea of having two distinct phases in the marginal likelihood curve.} This scheme for hyperparameter choice is illustrated in Fig. \ref{heuristic} using the simpler simulation scheme described in Subsection \ref{ss:illus}.

\section{Numerical experiments}

In this section, we perform some numerical experiments in order to highlight the main features of the proposed approach and to compare it with state-of-the-art methods.

\subsection{Simulation scheme}

\begin{figure}
	\centering
	\includegraphics[width=\columnwidth]{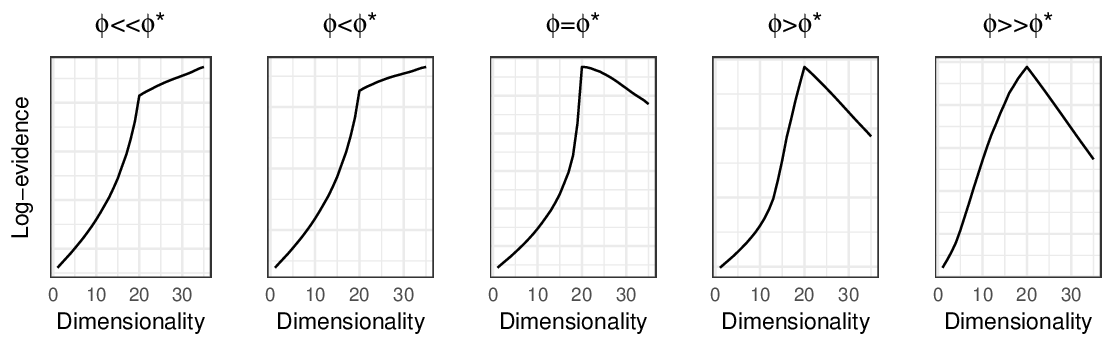}
	\caption[cap]{Different shapes of the marginal likelihood curve for growing values of $\phi$. $\phi^*$ corresponds to the maximum of the heuristic criterion that we describe is Subsection \ref{ss:hyp}. The true dimensionality is 20. More detailed results are available as an online GIF animation (\url{http://pamattei.github.io/animationeasy.gif}).}
	\label{heuristic}
\end{figure}

To assess the performance of our algorithm (referred hereafter as ngPPCA or NG, for short), we consider the following simulation scheme in the following experiments. We follow the simulation setup proposed in \cite{bouveyron2011intrinsic} based on their isotropic PPCA model. We therefore simulate data sets following the isotropic PPCA model which assumes that the covariance matrix of $X$ has only two different eigenvalues $\alpha$ and $\beta$ (instead of $d+1$ in the PPCA model). In this case, the signal-to-noise ratio (SNR hereafter) is simply defined by $$\textup{SNR}=\frac{\alpha d}{\beta (p-d)}.$$ In our simulation, $\beta$ is set up to 1 and $\alpha>1$, which will control the strength of the signal, varies to explore different signal-to-noise ratios. Then, an orthonormal $p\times p$ matrix $\mathbf{Q}$ is drawn uniformly at random. The data is eventually generated according to a centered Gaussian distribution with covariance matrix $$\mathbf{Q}^T\textup{diag}(\overbrace{\alpha,...,\alpha}^{d\textup{ times}},\overbrace{1,...,1}^{p-d\textup{ times}})\mathbf{Q}.$$ Finally, the number $p$ of variables is fixed to $50$ in all experiments and the number $n$ of observations varies in the range $\{40,50,70,100\}$.

\subsection{Introductory examples}
\label{ss:illus}

We first conduct two small simulations to illustrate the behaviour of our algorithm and its difference {with the Laplace approximations of \cite{minka2000automatic} and \citet{sobczyk2018}}. We consider two scenarios: a simple case and a harder and more realistic one.

\begin{figure}[p]
\centering
\includegraphics[width=0.9\columnwidth]{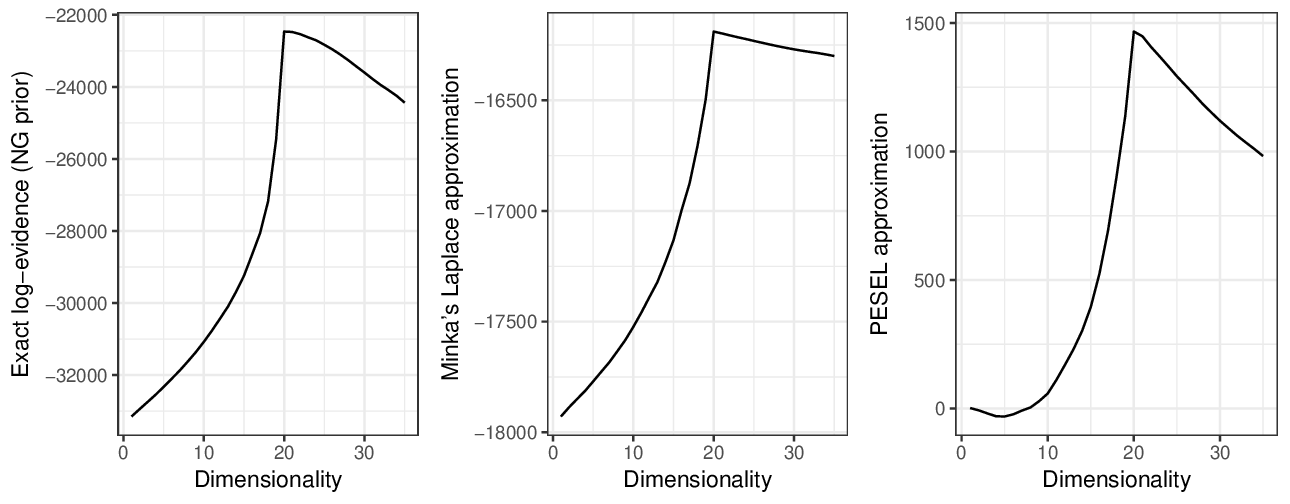}\\
\hspace{4ex} \small NG (our methodology) \hspace{8ex} Laplace \citep{minka2000automatic} \hspace{6ex} PESEL  \citep{sobczyk2018}
\caption{{Exact log-evidence for ngPPCA (left) and the Laplace approximations of \cite{minka2000automatic} (middle) and \citet{sobczyk2018} (right) for the simpler simulation scenario ($n=100$). The true dimensionality is $d=20$. All three curves have the desirable properties detailed is Subsection \ref{ss:hyp} and find the correct dimensionality $d=20$.}}
\label{easy}
\end{figure}

\begin{figure}[p]
\centering
\includegraphics[width=0.9\columnwidth]{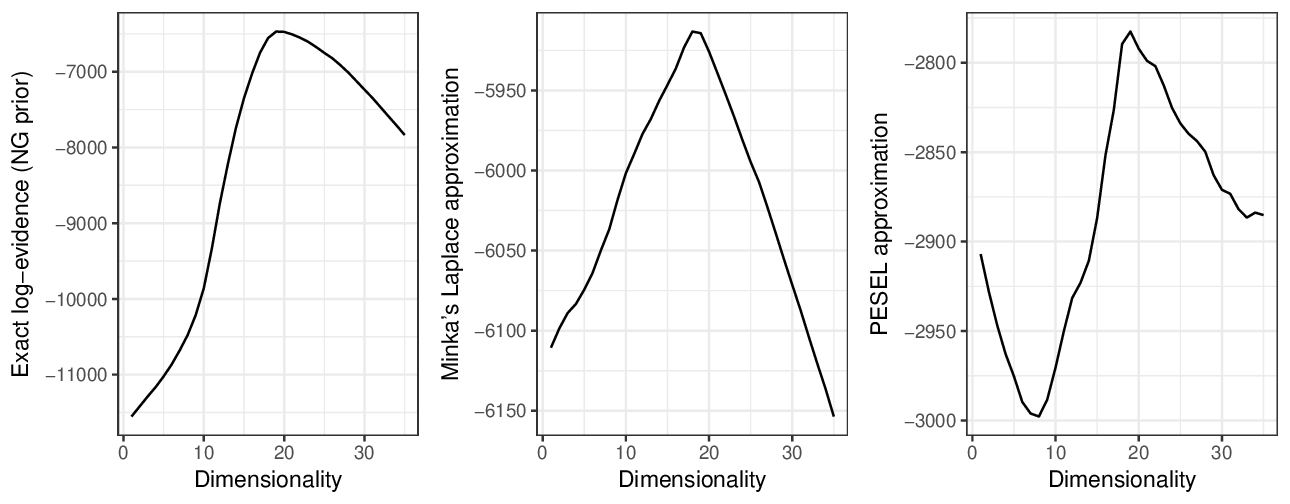}\\
\hspace{4ex} \small NG (our methodology) \hspace{8ex} Laplace \citep{minka2000automatic} \hspace{6ex} PESEL  \citep{sobczyk2018}
\caption{{Exact log-evidence for ngPPCA (left) and the Laplace approximations of \cite{minka2000automatic} (middle) and \citet{sobczyk2018} (right) for the more challenging simulation scenario ($n=40$). The true dimensionality is $d=20$. Both Laplace approximations have their way of preferring overly simple models, while the shape provided by ngPPCA is consistent with the ones obtained with $n=100$.}}
\label{hard}
\end{figure}

\paragraph{Simple scenario} We consider a setup with $n=100$ and $\textup{SNR}=20$. In this simple scenario, we first illustrate our heuristic for hyperparameter tuning by displaying marginal likelihood curves for different values of $\phi$ (Fig. \ref{heuristic}). The heuristic criterion allows to find the desired shape, leading to a correct dimensionality estimation. A GIF animation displaying all values of the criterion for a large grid of 200 values of $\phi$ is provided as a online material\footnote{\url{http://pamattei.github.io/animationeasy.gif}}. {This animation illustrates on this simple data set, a wide range of values of $\phi$ lead to dimensionality recovery. Our heuristic proposal $\phi^*$ is roughly located in the center of this range. On Fig. \ref{easy}, we compare the results of our algorithm with the Laplace approximations of the marginal likelihood of \cite{minka2000automatic} and \citet{sobczyk2018}. In this case, both methods recover the true dimensionality of the data and are very confident with their choice (the posterior probability of the true model is higher than $99\%$ for all approaches). The three curves have a similar shape, in compliance with the expected shape, as detailed in Subsection \ref{ss:hyp}.}

 \paragraph{Challenging scenario} We now consider a setup with $n=40$ and $\textup{SNR}=20$. A GIF animation illustrating hyperparameter tuning is provided online\footnote{\url{http://pamattei.github.io/animationhard.gif}}. Again, our results are compared with Laplace approximations (Fig. \ref{hard}).
{Regarding our exact approach (left panel), the marginal likelihood curve has a similar shape to the one of the first simulation. This shape is satisfactory, even though the algorithm slightly underestimates the dimensionality by choosing the model $\mathcal{M}_{19}$ (with posterior probability $>99\%$). The true model $\mathcal{M}_{20}$ is the second best model according to the NG prior.

Minka's \citeyearpar{minka2000automatic} Laplace approximation prefers simpler models (with lower intrinsic dimensionality). Indeed, the top two models chosen by this Laplace approximation are $\mathcal{M}_{18}$ (with posterior probability $73.8\%$) and $\mathcal{M}_{19}$ (with posterior probability $26.2\%$).

Like our approach, PESEL \citep{sobczyk2018} chooses $\mathcal{M}_{19}$ with posterior probability $>99\%$. However, it also has a tendency to prefer overly simple models. Indeed, the second best model is $\mathcal{M}_{18}$, and PESEL gives a surprisingly high posterior probability to very simple models with less than 5 dimensions. Indeed, for example, PESEL prefers $\mathcal{M}_{1},\mathcal{M}_{2},\mathcal{M}_{3},$ or $\mathcal{M}_{4}$ over $\mathcal{M}_{10}$. 

By being more resistant to underestimation, the exact method appears less likely to destroy valuable information, which would be damaging in a dimensionality selection context.

As a summary, those experiments confirm the expected behaviors of NG \textit{vs.} Laplace approximations: in the first scenario ($n=100, p=50$), the asymptotic assumption of the Laplace approximations are much more relevant than in the second setup ($n=40, p=50$). Our method, which does not rely on such an assumption, is less impacted by the reduction of the sample~size. Moreover, our heuristic for hyperparameter selection prevents against damaging underestimation.}

\subsection{Benchmark comparison with other dimension selection methods}

\begin{figure}
\centering
\includegraphics[width=0.72\textwidth]{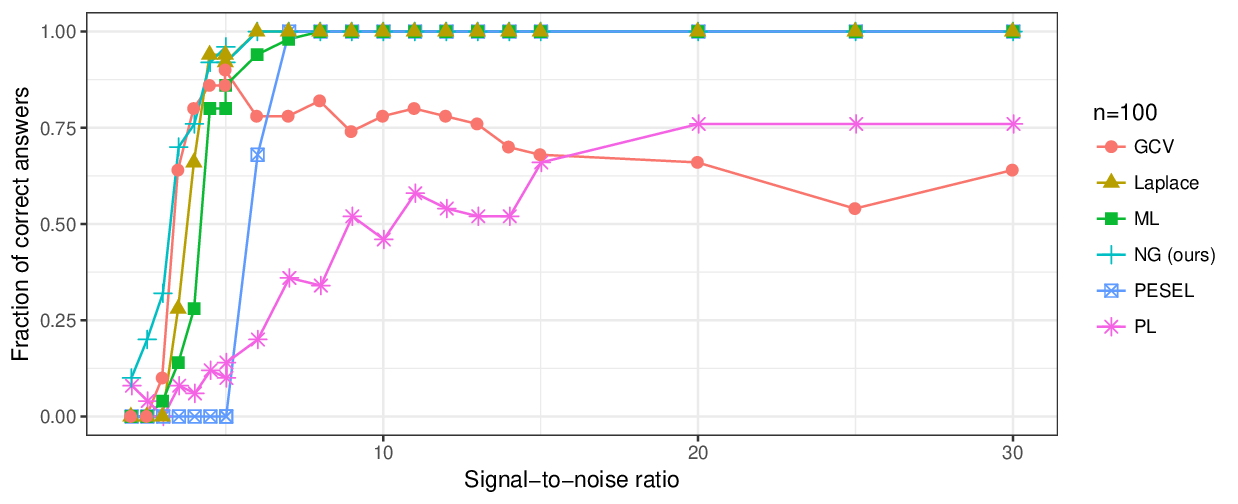}
\includegraphics[width=0.72\textwidth]{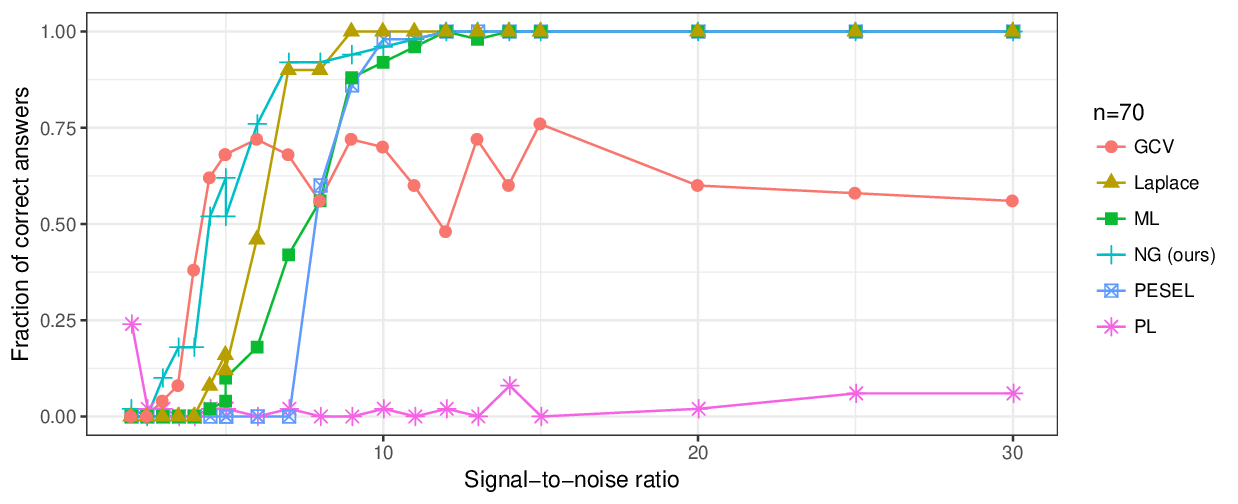}
\includegraphics[width=0.72\textwidth]{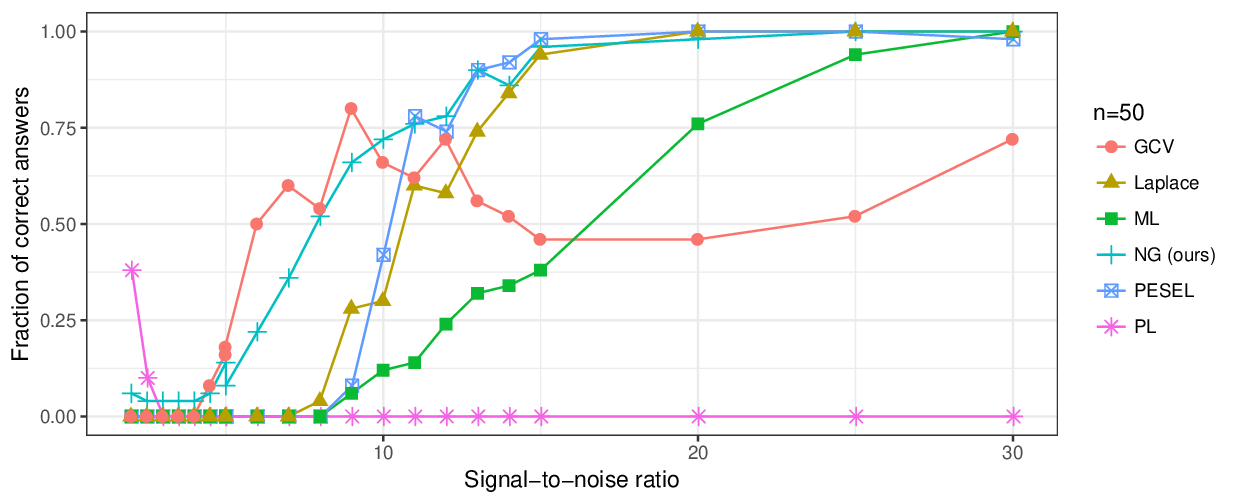}
\includegraphics[width=0.72\textwidth]{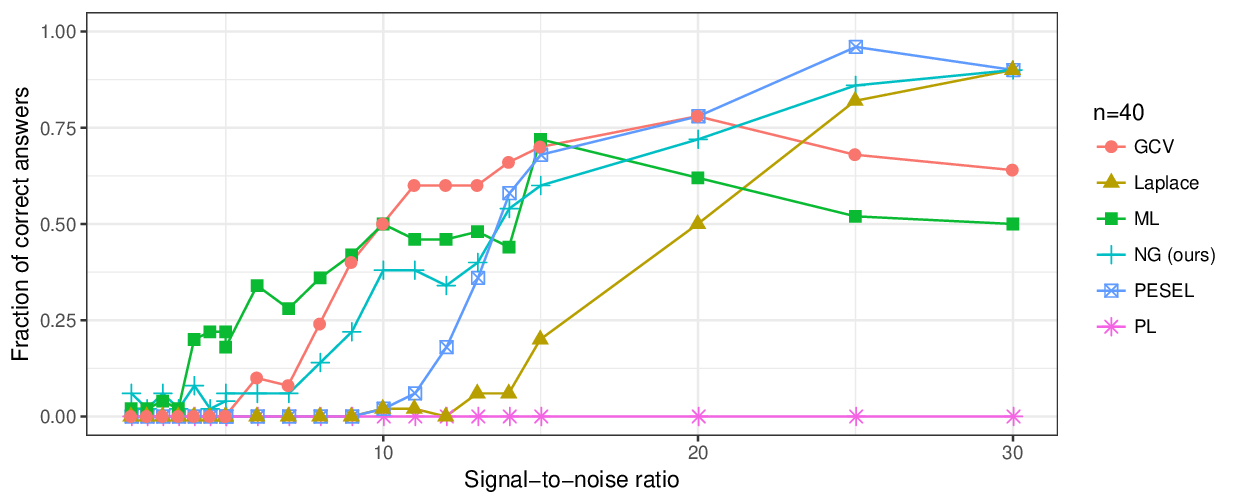}
\caption{Percentage of correctly estimated dimensions for different sample sizes (50 replications) of NGPPCA (NG) and its competitors for different signal-to-noise ratios. The true model corresponds to $d=20$. From top to bottom, the data sample sizes are respectively $100,~70,~50$ and $40$.}
\label{bench}
\end{figure}

\begin{figure}
	\centering
	\includegraphics[width=0.62\textwidth]{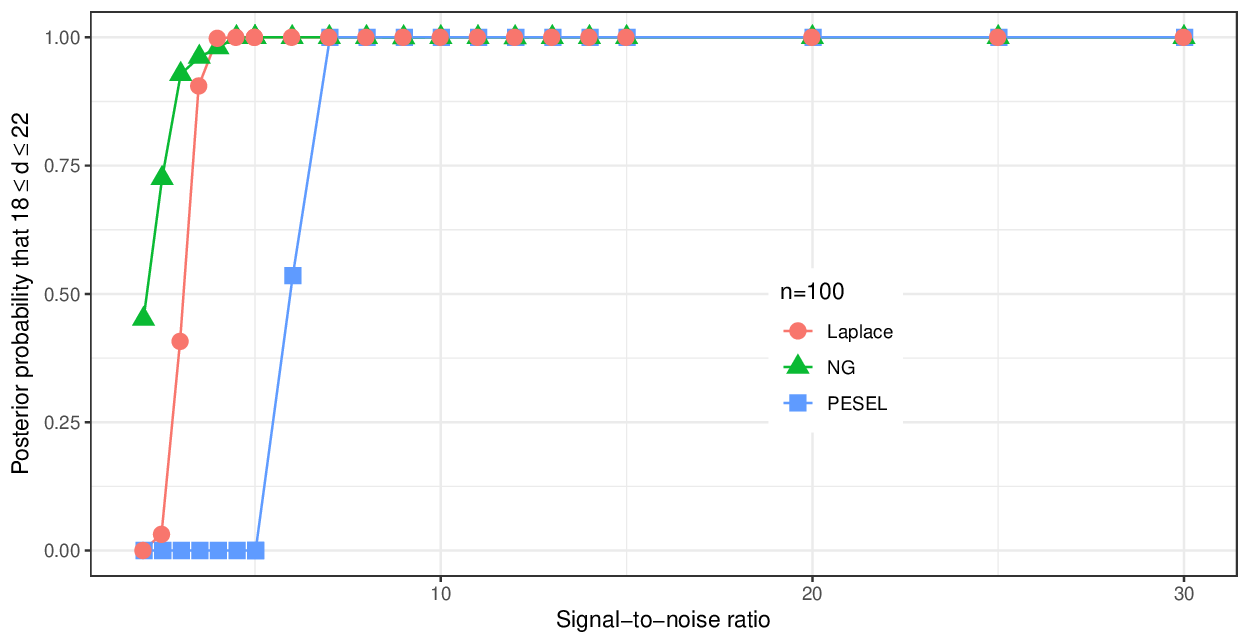}
	\includegraphics[width=0.62\textwidth]{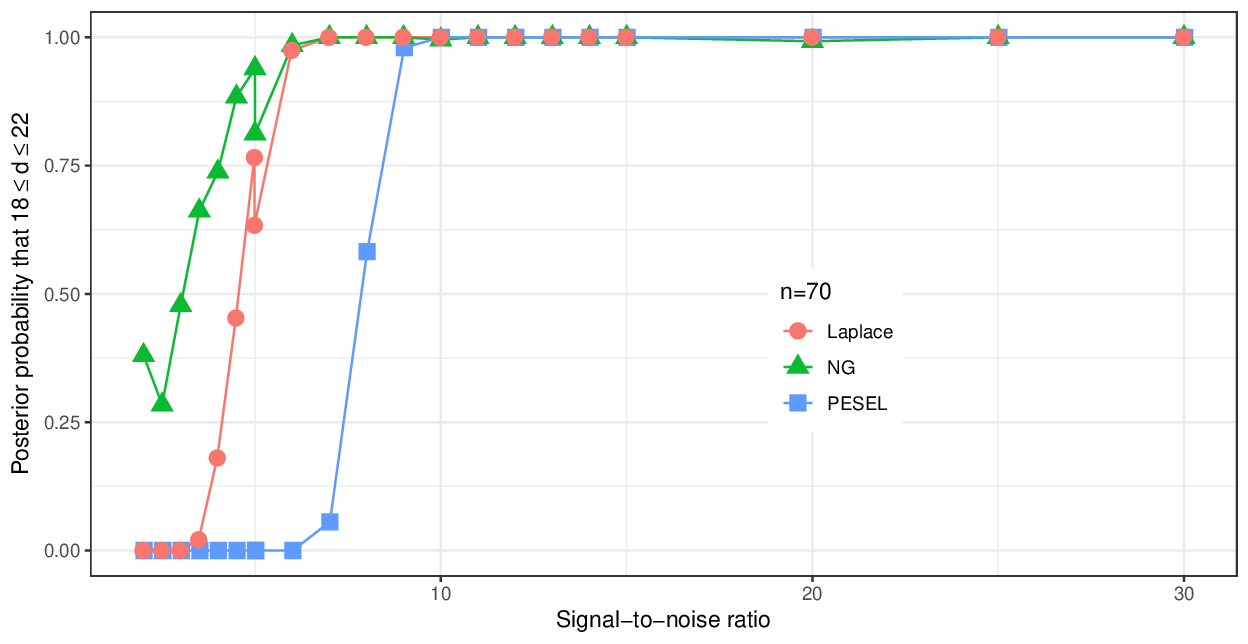}
	\includegraphics[width=0.62\textwidth]{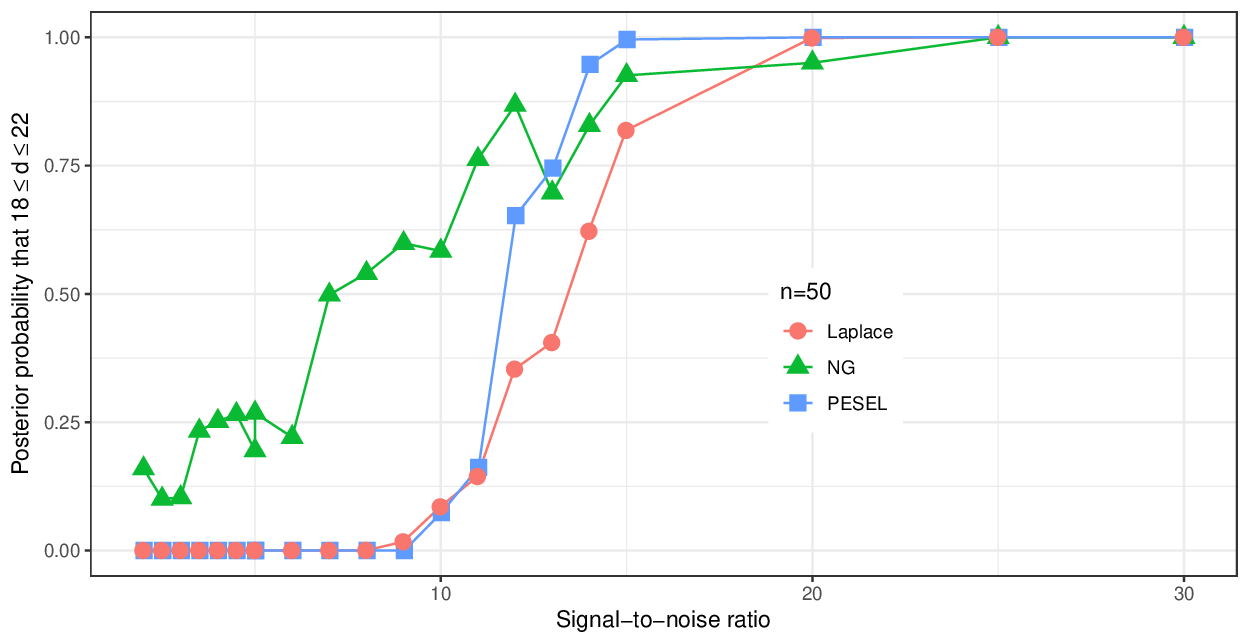}
	\includegraphics[width=0.62\textwidth]{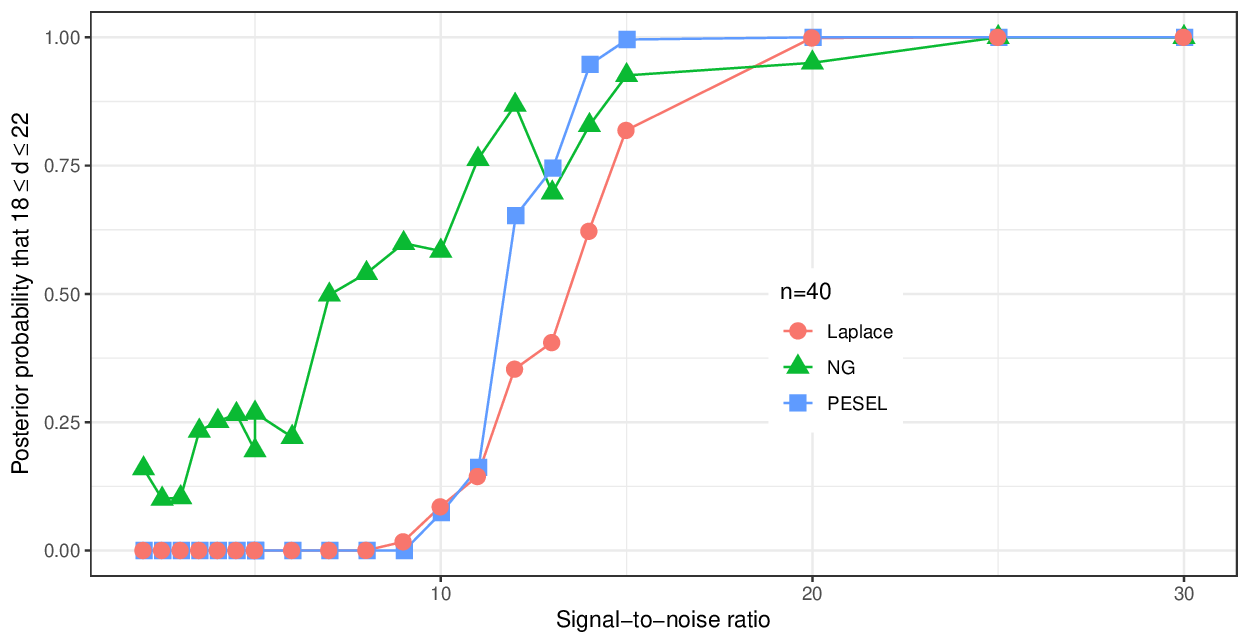}
	\caption{\modif{Posterior probability that $d \in \{18,...,22\}$ for different sample sizes (50 replications) of NGPPCA (NG) and its two Bayesian competitors for different signal-to-noise ratios. The true model corresponds to $d=20$. From top to bottom, the data sample sizes are respectively $100,~70,~50$ and $40$.}}
	\label{bench2}
\end{figure}

This section now focuses on the comparison of our methodology with other dimension selection methods. We here consider all possible scenarios with $n \in \{40,50,70,100\}$ and a SNR grid going from 1.5 to 30 (50 repetitions are made for each case). We compare the performance of our technique based on the normal-gamma prior (NG) with the following four competitors:
\begin{itemize}
\item the Laplace approximation of \cite{minka2000automatic} which is a benchmark Bayesian method for dimension selection,
\item the generalized cross-validation approximation (GCV) of \cite{josse2012selecting} which is known to give state of the art results in many scenarios (see the vast simulation study of \citealp{josse2012selecting}),
\item {the high-dimensional Laplace approximation of \citet{sobczyk2018} called PESEL, which performs well even in scenarios that imply a large number of variables,}
\item the profile likelihood approach (PL) of \cite{zhu2006automatic} which represents scree-based techniques and has been very popular in several different contexts \citep{fogel2007,evangelopoulos2012latent},
\item the ML approach of \cite{bouveyron2011intrinsic}, which maximizes a non-asymptotic criterion (the likelihood). Notice that this approach is specifically adapted to our simulation scheme and this advantage allows us to consider this technique as a gold-standard for this simulated data.
\end{itemize}
\modif{We use two metrics to evaluate the results, one based on point estimates of the dimensionality, and one based on the posterior mass of a neighbourhood of the true dimensionality.

	First, we assess in Fig. \ref{bench} the percentage of correct answers given by each algorithm, which is a standard measure used in other simulations studies (see e.g. \citealp{minka2000automatic,hoyle2008automatic,ulfarsson2008dimension}).}
One can first notice that all methods vastly outperform the profile likelihood (PL) approach, which seems not well-suited for small sample sizes. {Second, generalized cross-validation gives often satisfactory results, but fails to be competitive with model-based methods (Laplace, ML, NG and PESEL) when the SNR is high. The ML approach has a good behavior, especially when $n$ is very small, this is partly explained by the fact that it is designed for this very simulation setup.
Regarding the three Bayesian methods, the Laplace approximation is often outperformed by PESEL, and consistently outperformed by our approach (NG), mainly because of the important violation of the $n \rightarrow +\infty$ assumption. PESEL gives very good results at high SNR, but is outperformed by our approach at low SNR, which eventually
is the only method that gives satisfactory results in all settings (high and low SNR, moderate and small $n$).}

\modif{Second, to compare the relevance of the various Bayesian posteriors, we evaluate the posterior probabilities that $d \in \{18,...,22\}$ for the two Laplace approximations, as well as for our approach (Fig. \ref{bench2}). By this standard, NG outperforms the two Laplace approximations in almost all scenarios. Perhaps more importantly, these results suggest that both Laplace approximations are generally overly confident, and underestimate model uncertainty. This is especially the case for PESEL, which goes very quickly from being very confident that $d$ is far away from 20, to being very confident that $d$ is exactly 20. On the other hand, our approach always gives a small posterior mass to the fact that $d$ is close to 20, and slowly grows more and more confident. These results illustrate that BIC-like approximations usually provide good point estimates, but poor posterior estimates (see e.g. \citealp[Section 5.2]{drton2017}, for another example of this quite general phenomenon).}

\section{Conclusion}

PCA is more of a descriptive and exploratory tool than a model. Therefore, no unique dimension selection method should be uniquely preferred -- sometimes, very relevant information may actually lie within the \emph{last} PCs \citep[section 3.4]{jolliffe2002principal}.

However, PCA's ubiquity in the statistical world makes necessary the search for guidance procedures to help the practitioner choose the number of PCs. This need is even more critical when the data are scarce or particularly expensive. Our work, by deviating from usually adopted asymptotic settings, is a step in that direction. Regarding future work, our exact computation of model posterior probabilities may be used to perform Bayesian model averaging \citep{hoeting1999bayesian} in predictive contexts. Potential applications could involve principal component regression \citep[Chapter 8]{jolliffe2002principal}, image denoising \citep{deledalle2011image}, or deep learning \citep{pcanet}. {In that context, potential drawbacks of approaches based on the marginal likelihood (like ours or the Laplace approximations) would be that they can suffer importantly from model misspecification, and that they might not be optimal for predictive purposes.}

As a concluding note, this work comes as an illustration that exactly computing the marginal likelihood is sometimes easier than expected. Although both recent asymptotic approximations \citep{drton2017} and the MCMC arsenal \citep{friel2012estimating} are well-equipped to deal with marginal likelihoods, we argue, like \cite{lin2009marginal}, that finding exact expressions is an important task that should not be deemed untractable too hastily.

\section*{Acknowledgement}

Part of this work was conducted while Pierre-Alexandre Mattei was visiting University College Dublin, funded by the Fondation Sciences Mathématiques de Paris (FSMP).

\bibliography{biblio}

\end{document}